\documentclass{llncs}

\usepackage{graphicx}
\usepackage{hyperref}
\usepackage{paralist}

\def\etal{{\it et al.}}
\def\ie{{\it i.e.}}
\def\eg{{\it e.g.}}
\def\opt{{\mbox{\sc opt}}}
\def\eps{\varepsilon}
\def\polylog{\mathop{\rm polylog}}

\def\DDD{\mathcal{D}}
\def\LLL{\mathcal{L}}
\def\RRR{\mathcal{R}}
\def\SSS{\mathcal{S}}

\newcommand{\BIGO}[1]{\mathcal{O}\kern-2pt\left(#1\right)}
\newcommand{\LITTLEO}[1]{{o}\kern-2pt\left(#1\right)}
\newcommand{\THETA}[1]{\Theta\kern-2pt\left(#1\right)}
\newcommand{\ceil}[1]{\left\lceil #1 \right\rceil}
\newcommand{\hide}[1]{}

\begin{document}

\pagestyle{headings}

\title{Fast Vertex Guarding for Polygons}
\author{James King}
\institute{Department of Physics\\University of Oxford}

\maketitle

\begin{abstract}
For a polygon $P$ with $n$ vertices, the vertex guarding problem asks for the minimum subset $G$ of $P$'s vertices such that every point in $P$ is seen by at least one point in $G$.  This problem is NP-complete and APX-hard.  The first approximation algorithm (Ghosh, 1987) involves decomposing $P$ into $\BIGO{n^4}$ cells that are equivalence classes for visibility from the vertices of $P$.  This discretized problem can then be treated as an instance of set cover and solved in $\BIGO{n^5}$ time with a greedy $\BIGO{\log n}$-approximation algorithm.  Ghosh (2010) recently revisited the algorithm, noting that minimum visibility decompositions for simple polygons (Bose \etal, 2000) have only $\BIGO{n^3}$ cells, improving the running time of the algorithm to $\BIGO{n^4}$ for simple polygons.

In this paper we show that, since minimum visibility decompositions for simple polygons have only $\BIGO{n^2}$ cells of \emph{minimal} visibility (Bose \etal, 2000), the running time of the algorithm can be further improved to $\BIGO{n^3}$.  This result was obtained independently by Jang and Kwon (2011).  We extend the result of Bose \etal~to polygons with holes, showing that a minimum visibility decomposition of a polygon with $h$ holes has only $\BIGO{(h+1)n^3}$ cells and only $\BIGO{(h+1)^2n^2}$ cells of minimal visibility.  We exploit this result to obtain a faster algorithm for vertex guarding polygons with holes.  We then show that, in the same time complexity, we can attain approximation factors of $\BIGO{\log\log\opt}$ for simple polygons and $\BIGO{(1+\log{(h+1)}) \log\opt}$ for polygons with holes.
\end{abstract}

\clearpage

\setcounter{page}{1}

\section{Introduction}

Art gallery problems, \ie, polygon guarding problems, are motivated by the question, ``How many security cameras are required to guard an art gallery?''  The art gallery is modeled as a connected polygon $P$.  A camera, which we henceforth call a \emph{guard}, is modeled as a point in the polygon, and we say that a guard $g$ \emph{sees} a point $q$ in the polygon if the line segment $\overline{gq}$ is contained in $P$.  The visibility polygon of a point $p$, denoted $\mathrm{Vis}(p)$, is the set of points in $P$ that see $p$.  We call a set $G$ of points a \emph{guarding set} if every point in $P$ is seen by some $g\in G$, \ie, if $\bigcup_{g\in G} \mathrm{Vis}(g) = P$.  Let $V(P)$ denote the vertex set of $P$ and let $\partial P$ denote the boundary of $P$.  We assume that $P$ is closed and non-degenerate so that $V(P) \subset \partial P \subset P$.

We consider the minimization problem that asks, given an input polygon $P$ with $n$ vertices, for a minimum guarding set for $P$.  Variants of this problem typically differ based on what points in $P$ must be guarded and where guards can be placed, as well as whether $P$ is simple or contains holes.  Typically we want to guard either $P$ or $\partial P$, and our set of potential guards is typically $V(P)$ (vertex guards), $\partial P$ (perimeter guards), or $P$ (point guards).  This paper concerns the variant in which we must guard all of $P$ from vertices of $P$.  For results on art gallery problems not related to minimization problems we direct the reader to O'Rourke's book \cite{orourke1987}, which is available for free online.

\subsection{Related Work}

\subsubsection{Hardness Results}

The problem was proved to be NP-complete for polygons with holes by O'Rourke and Supowit \cite{orourke1983}.  For guarding simple polygons it was proved to be NP-complete for vertex guards by Lee and Lin \cite{lee1986}; their proof was generalized to work for point guards by Aggarwal \cite{aggarwal1984}.  This raises the question of approximability.  There are two major hardness results.  First, for guarding simple polygons, Eidenbenz \cite{eidenbenz1998} proved that the problem is APX-complete, meaning that we cannot do better than a constant-factor approximation algorithm in polynomial time unless $\mathrm{P=NP}$.  Subsequently, for guarding polygons with holes, Eidenbenz \etal\ \cite{eidenbenz2001} proved that the minimization problem is as hard to approximate as set cover{} in general if there is no restriction on the number of holes.  It therefore follows from results about the inapproximability of set cover{} \cite{raz1997,feige1998,alon2006} that, for polygons with an unbounded number of holes, it is NP-hard to find a guarding set of size $o(\log n)$.  These hardness results hold whether we are dealing with vertex guards, perimeter guards, or point guards.

\subsubsection{Approximation Algorithms}
Ghosh \cite{ghosh1987} provided an $O(\log n)$-approximation algorithm for guarding polygons with or without holes with vertex guards.  His algorithm decomposes the input polygon into a polynomial number of cells such that each point in a given cell is seen by the same set of vertices.  This discretization allows the guarding problem to be treated as an instance of set cover and solved using general techniques.  In fact, applying methods for set cover developed after Ghosh's algorithm, it is easy to obtain an approximation factor of $O(\log\opt)$ for vertex guarding simple polygons or $O(\log h \log \opt)$ for vertex guarding a polygon with $h$ holes.

When considering point guards or perimeter guards, discretization is far more complicated since two distinct points will not typically be seen by the same set of potential guards even if they are very close to each other.  Deshpande \etal{}\null{} \cite{deshpande2007} obtain an approximation factor of $O(\log\opt)$ for point guards or perimeter guards by developing a sophisticated discretization method that runs in pseudopolynomial time\footnote{It is a pseudopolynomial-time algorithm in that its running time may be linear in the ratio between the longest and shortest distances between two vertices.}.  Efrat and Har-Peled \cite{efrat2006} provided a randomized algorithm with the same approximation ratio that runs in fully polynomial expected time; their discretization technique involves only considering guards that lie on the points of a very fine grid.

\subsection{Range Spaces and Discretization}

Guarding problems can naturally be expressed as instances of set cover or hitting set.  We wish to model an instance of a guarding problem as an instance of hitting set on a range space $\SSS = (X,\RRR)$, constructed as follows.  $X$ is equal to the set $S_G$ of potential guard locations.  For each point $p$ that needs to be guarded, $R_p$ is the set of potential guards that see $p$.  Now $\RRR = \{R_p:p\in S_T\}$, where $S_T$ is the set of points that must be guarded.  For the vertex guarding problem, $S_G=V(P)$ and $S_T = P$.

We assume $S_G$ is finite; in our case $|S_G| = n$.  If $S_T$ is \emph{not} finite, \eg, when $S_T=P$, we need to discretize it.  The goal of discretization is to find a finite representative subset $S_T'\subset S_T$ such that any subset of $S_G$ that guards $S_T'$ also guards $S_T$.  With such a set we are able to forget about $S_T$ and focus on the finite range space $\left(S_G, \{R_p:p\in S_T'\} \right)$ induced by $S_G$ and $S_T'$.

We consider general techniques for solving hitting set for finite range spaces.  For a finite range space $\SSS = (X,\RRR)$ the time complexity typically depends on $|X|$ (\ie, the number of elements) and $|\RRR|$ (\ie, the number of ranges).

\subsection{Visibility Decompositions for Polygons}

For points $p,q\in S_T$, we say that $p$ and $q$ are equivalent if and only if $R_p=R_q$.  In a decomposition of a polygon into cells, we say that a cell is an \emph{equivalence cell} if all points are equivalent.  A natural discretization strategy is to partition $S_T$ into a finite number of sets that are closed under equivalence, and then to build a subset $S_T'$ by taking one representative point from each set.  A subset of $S_G$ guards $S_T$ if and only if it guards $S_T'$.

Ghosh \cite{ghosh1987} did this for the vertex guarding problem in which $S_G =V$ and $S_T = P$.  His algorithm decomposes the input polygon into a polynomial number of cells such that each point in a given cell is seen by the same set of vertices.  For two distinct vertices that see each other, consider the line through them.  The set of all such lines decomposes $P$ into a number of equivalence cells (see Figure \ref{fig:decomposition_1}).  Using a simple inequality for general line arrangements it can be seen that the number of cells is $\BIGO{n^4}$.  Ghosh originally used his discretization technique, along with the greedy set cover approximation algorithm (see, \eg, \cite[pp.~16--19]{vazirani2001}), to provide a $\BIGO{\log n}$-approximation algorithm.  This decomposition can be generalized to work for any finite set $S_G$ of potential guards---simply shoot a ray from every point $p\in S_G$ through every vertex seen by $p$.  In each of the $\BIGO{|S_G|^2n^2}$ resulting cells, any two points see the same subset of $S_G$.

\begin{figure}
\begin{minipage}[t]{0.46\linewidth}
\centering
\includegraphics[width=\columnwidth]{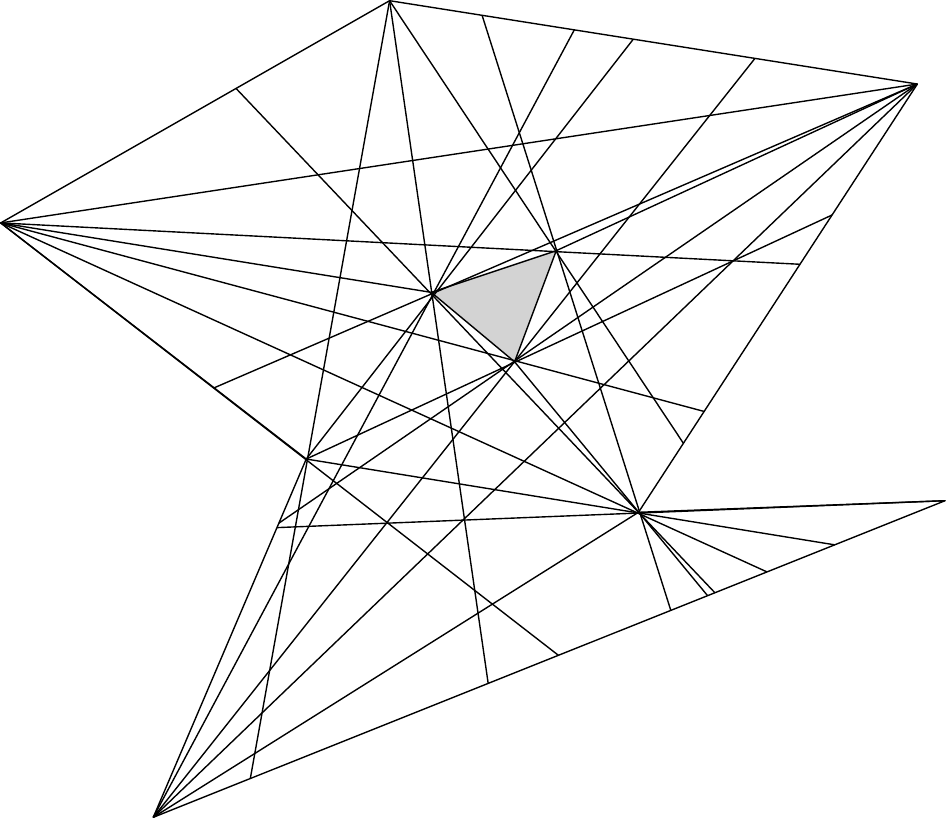}
\caption{A polygon decomposed into $\BIGO{n^4}$ cells by cutting along any line passing through two vertices that see each other.  The shaded region is a hole.\label{fig:decomposition_1}}
\end{minipage}
\hspace{0.05\linewidth}
\begin{minipage}[t]{0.46\linewidth}
\centering
\includegraphics[width=\columnwidth]{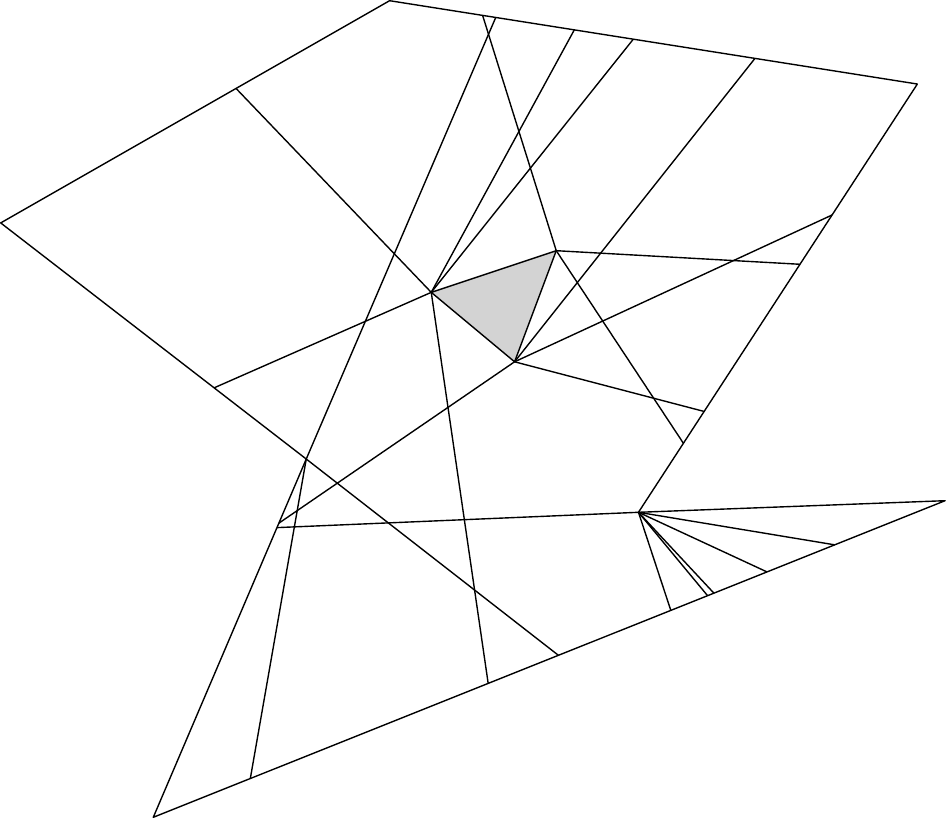}
\caption{A polygon decomposed into the $\BIGO{n^3}$ cells of its minimum visibility decomposition $\DDD_V(P,V(P))$.  The shaded region is a hole.\label{fig:decomposition_2}}
\end{minipage}
\end{figure}

Bose \etal{} \cite{bose2002} introduced a minimum decomposition with fewer cells (some of their results were obtained independently by Guibas \etal~\cite{guibas1997}, but we focus on the results as stated and proved by Bose \etal).  Let the \emph{visibility decomposition}, denoted $\DDD_V(P,S_G)$, be the minimum decomposition of a polygon into equivalence cells with regard to $S_G$.  It is minimum in that the union of cell boundaries in this decomposition is exactly equal to $\bigcup_{p\in S_G}\partial(\mathrm{Vis}(p))$.  The decomposition can be constructed as follows.  For a point $p\in S_G$ that sees a vertex $v$, instead of cutting along the entire ray shot from $p$ through $v$, we leave the line segment $\overline{pv}$ and cut only from $v$ until we hit $\partial P$ (see Figure \ref{fig:decomposition_2}).  Bose \etal\ call such a cut a \emph{window of point} $p$ and they note that the boundary of $\mathrm{Vis}(p)$ consists only of windows of $p$ and parts of $\partial P$.  They proved an upper bound of $\BIGO{n^3}$ for the number of cells in $\DDD_V(P,V(P))$ for a simple polygon $P$.

Ghosh \cite{ghosh2010} recently revisited his algorithm, using these minimum visibility decompositions to improve the running time.  His updated algorithm guarantees an approximation factor of $\BIGO{\log n}$ and runs in $\BIGO{n^4}$ for simple polygons and $\BIGO{n^5}$ for polygons with holes.

\subsection{Novel Contributions}

In this paper we exploit another result of Bose \etal~\cite{bose2002}, namely the fact that, while minimum visibility decompositions for simple polygons can have $\THETA{n^3}$ cells, they have only $\BIGO{n^2}$ cells of \emph{minimal} visibility.  We use this to further improve the running time of Ghosh's algorithm from $\BIGO{n^4}$ to $\BIGO{n^3}$.

We then extend the result of Bose \etal~to polygons with $h$ holes, parameterizing bounds not only by $n$ but also by $h$.  We show that a minimum visibility decomposition of a polygon with $h\geq 0$ holes has only $\BIGO{(h+1)n^3}$ cells (Theorem \ref{thm:cells}) and only $\BIGO{(h+1)^2n^2}$ cells of minimal visibility (Theorem \ref{thm:sinks}).  We exploit this result to improve the running time of Ghosh's algorithm from $\BIGO{n^5}$ to $\BIGO{(h+1)^2n^3}$; this is a strict asymptotic improvement for $h = \LITTLEO{n}$.

Having presented algorithms with faster running times, we turn our attention to improving the approximation factor.  We show that, with the same time complexity bound of $\BIGO{(h+1)^2n^3}$, we can apply standard random sampling techniques for range spaces of bounded VC dimension to obtain an approximation factor of $\BIGO{(1+\log{(h+1)}) \log\opt}$ for polygons with $h\geq 0$ holes (Theorem \ref{thm:main-withholes}).

Finally, we show that, with the same time bound of $\BIGO{n^3}$, we can achieve an approximation ratio of $\BIGO{\log\log\opt}$ for simple polygons using the improved $\eps$-net finders of King and Kirkpatrick \cite{king2010loglog} (Theorem \ref{thm:main-loglog}).

The author originally suggested the idea of exploiting the number of sinks to improve running time in a recent thesis \cite{king2010thesis}.  This idea was also used independently by Jang and Kwon \cite{jang2011}, who obtained the same time complexity as us for simple polygons.  Jang and Kwon also consider the problem of edge guards, whereas we do not.  However, they do not consider improvement of the $\BIGO{\log n}$ approximation factor or consider polygons with holes.


\subsection{Model of Computation}

We assume the real-RAM model of computation \cite{preparata1985}.  We also assume that the polygon and any holes are non-degenerate, and that vertices of the polygon are in general position, \ie, no three are collinear.

\section{Polygon Decompositions}

\subsection{Notation}

For a polygon $P$ with $h\geq 0$ holes, $\overline P \cup \partial P$ is the complement of the relative interior of $P$, or the closure of the complement of $P$.  We use $C_0,\dots,C_h$ to denote the $h+1$ components of $\overline P \cup \partial P$.  $C_0$ is the polygon's exterior and $C_1,\dots,C_h$ are the holes of the polygon.

In addition to left windows, right windows, left pockets, and right pockets defined by Bose \etal{}~for simple polygons \cite{bose2002} (see Figure \ref{fig:windows_and_pockets}), we have T-windows, which are trans-component windows.  A window is a T-window if its two endpoints are on different components.  There are no pockets associated with T-windows.

\begin{figure}
\centering
\includegraphics[height=6cm]{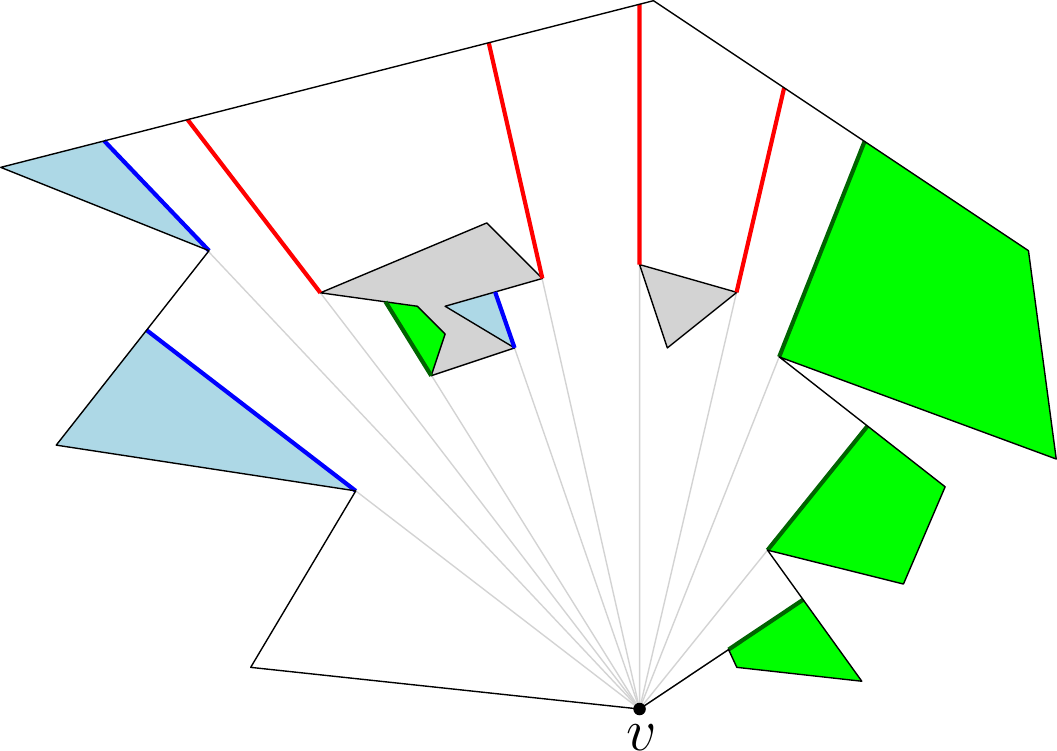}
\caption{Windows and pockets of a vertex $v$.  Windows are represented by coloured line segments: left windows are blue, right windows are green, and T-windows are red.  Left and right pockets are shaded blue and green respectively.\label{fig:windows_and_pockets}}
\end{figure}

A window of a vertex $v$ has two endpoints; the endpoint closer to $v$ is called the \emph{base} and the endpoint farther from $v$ is called the \emph{end}.  The \emph{window half plane} of a left (resp. right) window of $v$ whose base is $b$ is the half plane on the left (resp. right) side of the oriented line from $v$ to $b$.

\subsection{Bounding the Number of Cells}

Here we extend the results of Bose \etal{}~\cite{bose2002} concerning the number of cells in a visibility decomposition.  Our proofs are, to a great degree, extensions of their proofs.

\begin{lemma}[\cite{bose2002}]\label{lem:bose1}
No point $z$ in a pocket $Q$ of a window $w$ is visible to any point $y$ that is inside $w$'s half plane but outside $Q$.
\end{lemma}
\begin{proof}
The proof of \cite[Lemma 1]{bose2002} works without modification when generalizing to polygons with holes.  For completeness we restate it here.

Let $w$ have $b$ as its base and $e$ as its end.  Since $Q\cup w$ is a polygon, the line segment $\overline{zy}$ intersects $Q\cup w$ at least once.  As $\overline{yz}$ cannot intersect $\overline{be}$, it intersects some other line segment on the boundary of the polygon.  So, $y$ and $z$ are not visible with respect to the chain $Q$.
\end{proof}

\begin{lemma}\label{lem:ds_sequence}
For any vertex $v$ there are at most $2h$ T-windows of $v$.
\end{lemma}
\begin{proof}
Consider the \emph{component sequence} of a vertex $v$ obtained as follows.  Let $r$ be a ray emanating from $v$ that bisects the external angle at $v$ (or, if $v$ is a vertex of a hole, the angle interior to the hole).  From its starting position, the $r$ rotates clockwise around $v$, making one full rotation.  Now consider the first component that is hit by $r$.  As $r$ rotates, the component hit changes exactly where there is a $T$-window of $v$.  The changing sequence of components hit by $r$ through one full rotation is the \emph{component sequence} of $v$.  This sequence starts and ends with $P_E$.

It remains to show that the component sequence of $v$ is a Davenport-Schinzel sequence of order 2; since it is a sequence over $h+1$ symbols, this implies that the sequence length is at most $2h+1$ \cite{sharir1995}, which in turn implies that there are at most $2h$ T-windows of $v$.  A Davenport-Schinzel sequence of order 2 is a sequence in which, for any two symbols $a$ and $b$, the subsequence $\dots,a,\dots,b,\dots,a,\dots,b,\dots$ does not appear.  

The rest of the proof follows a simple geometric argument illustrated in Figure \ref{fig:davenport_schinzel}.  Assume that, in the rotation of $r$, it hits $C_i$ at point $u_1$, then hits $C_j$, then hits $C_i$ at point $u_2$.  Then the union of $C_i$ and line segments $\overline{vu_1}$ and $\overline{vu_2}$, there is a Jordan curve that contains the relative interior of $C_j$ in its interior region.  Thus $C_j$ will not be hit by $r$ in its rotation after it passes $u_2$.
\qed
\end{proof}

\begin{figure}
\begin{minipage}[t]{0.36\linewidth}
\centering
\includegraphics[height=5cm]{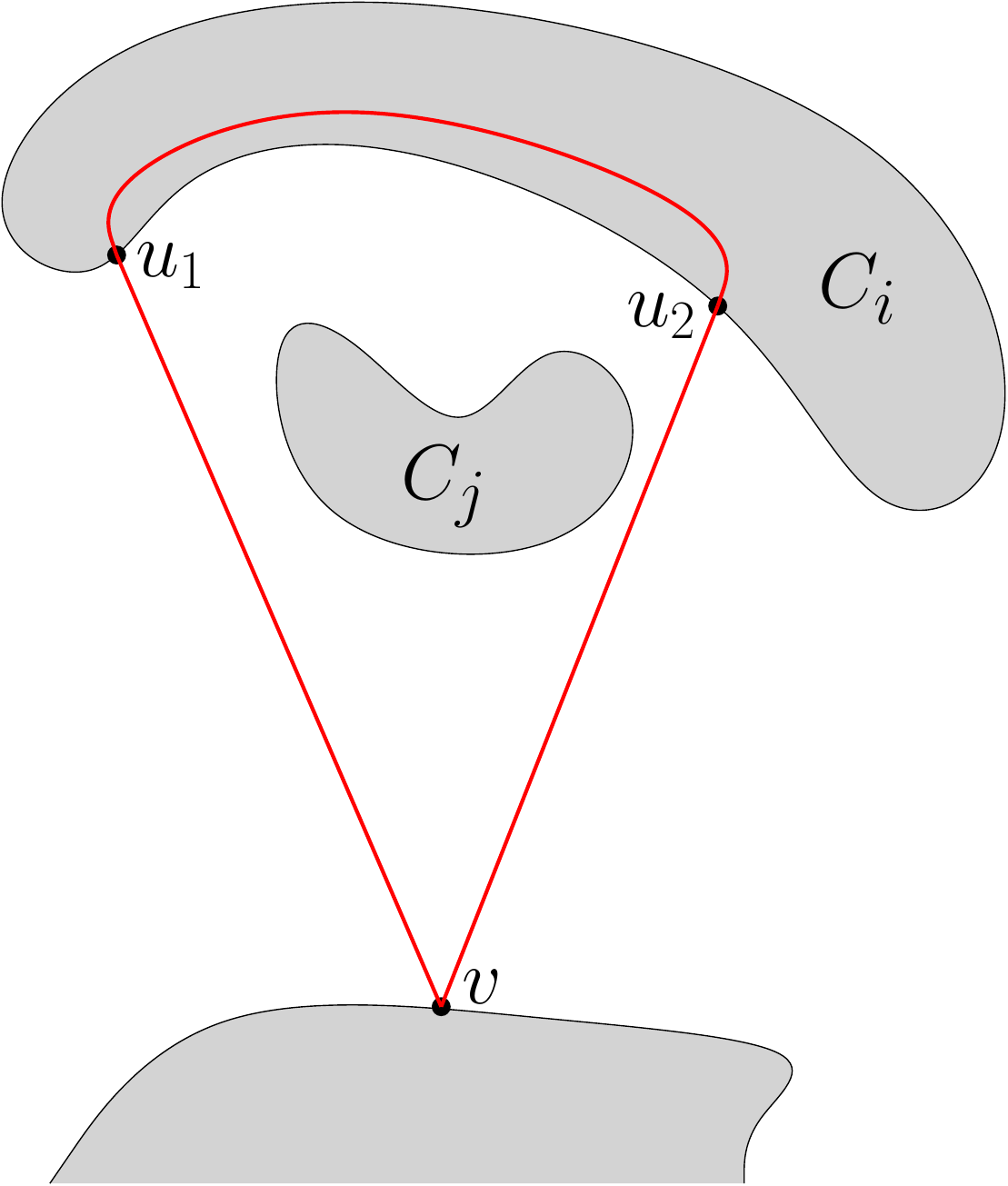}
\caption{If, in rotational order, $v$ sees $C_i$, then $C_j$, then $C_i$ again, then it will never see $C_j$ after that because $C_j$ is `trapped' in the interior of the red Jordan curve.\label{fig:davenport_schinzel}}
\end{minipage}
\hspace{0.05\linewidth}
\begin{minipage}[t]{0.56\linewidth}
\centering
\includegraphics[height=5cm]{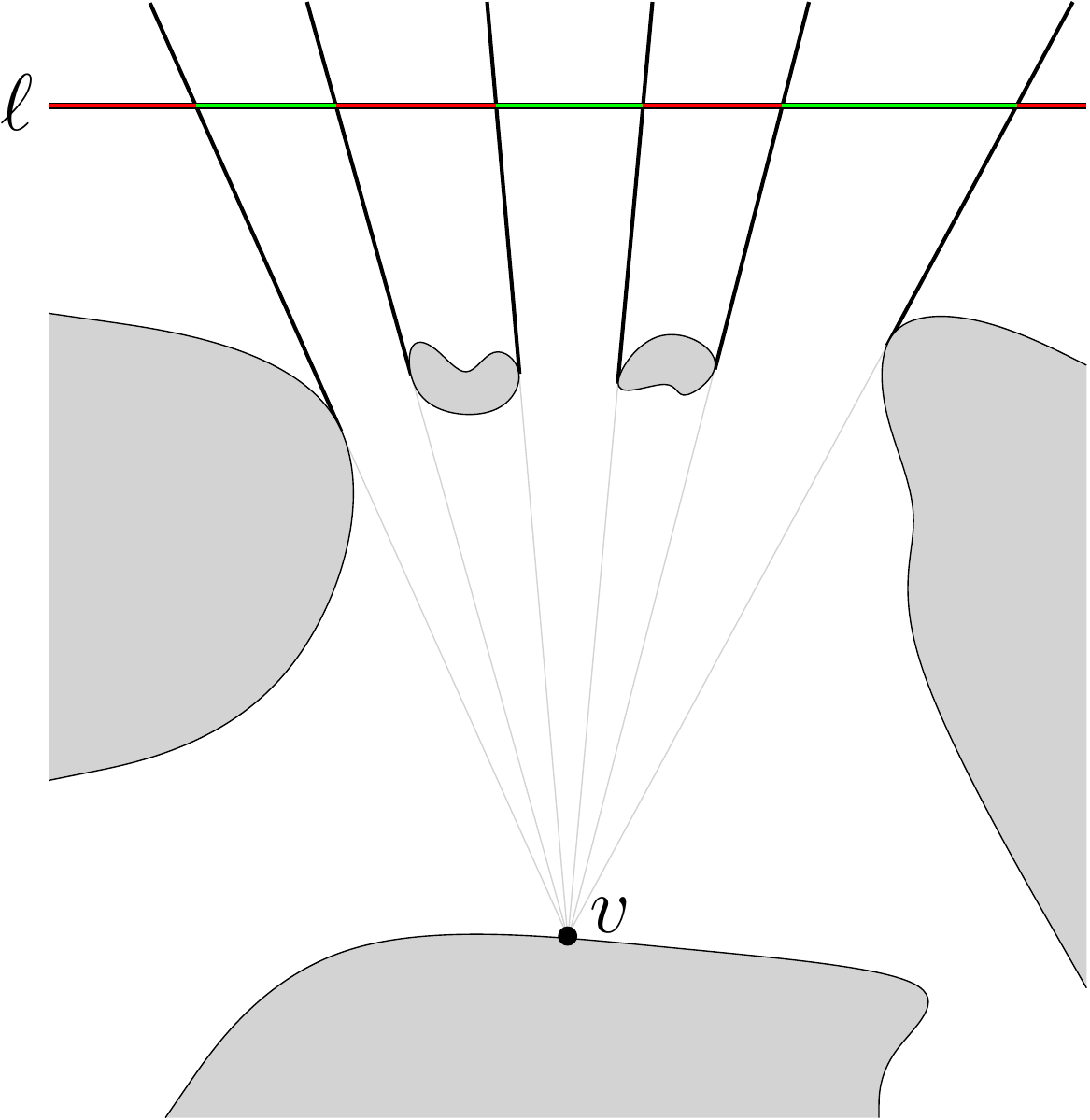}
\caption{An illustration accompanying the proof of Lemma \ref{lem:windows_intersecting_segment}.  The line segment $\ell$ is intersected by 6 windows of $v$.  Alternating invervals of visibility from $v$ are shown in red and green.  For $\ell$ to be intersected by more windows of $v$, more than 2 holes would be required.\label{fig:windows_intersecting_segment}}
\end{minipage}
\end{figure}

\begin{lemma}\label{lem:windows_intersecting_segment}
For a polygon $P$ with $h$ holes, a fixed vertex $v$, and a fixed line segment $\ell$ in $P$, $\ell$ intersects at most  $2(h+1)$ windows of $v$: one right window, one left window, and $2h$ T-windows.
\end{lemma}
\begin{proof}
See Figure \ref{fig:windows_intersecting_segment}.  Of the windows of $p$ that cross $\ell$, at most one is a right window and at most one is a left window.  To see this, consider a point $p$ moving from one endpoint of $\ell$ to the other, directed so that line segment $\overline{vp}$ rotates clockwise around $v$ as $p$ moves.  $p$ can exit at most one left pocket of $v$ and cannot enter a left pocket of $v$.  $p$ can enter at most one right pocket of $v$ and cannot exit a right pocket of $v$.  At most $2h$ T-windows of $p$ cross $\ell$ because, by Lemma \ref{lem:ds_sequence}, there are at most $2h$ T-windows of $p$.
\qed\end{proof}

\begin{theorem}\label{thm:cells}
For a polygon $P$ with $h$ holes, $\DDD_V(P,V(P))$ contains $\BIGO{(h+1)n^3}$ cells.
\end{theorem}
\begin{proof}
By Lemma \ref{lem:windows_intersecting_segment}, each window is crossed by at most $2(h+1)$ windows of each vertex, and thus $2(h+1)n$ other windows in total.  The total number of windows is $\BIGO{n^2}$, so the number of points at which windows cross is $\BIGO{(h+1)n^3}$.  The theorem follows by applying Euler's formula for planar graphs.
\qed\end{proof}

\subsection{Constructing and Using the Decomposition}

Chazelle and Edelsbrunner \cite{chazelle1992} give an efficient algorithm for finding the planar decomposition defined by a set $\LLL$ of line segments.  The running time is $\BIGO{|\LLL|\log |\LLL| + k}$, where $k$ is the number of intersections.  In our case, $|\LLL|$ is $\BIGO{n^2}$ and Theorem \ref{thm:cells} tells us that $k$ is $\BIGO{(h+1)n^3}$, so the total runtime for this step is $\BIGO{(h+1)n^3}$.

The visibility decomposition $\DDD_V(P,V(P))$ has an associated planar graph $G$.  We are particularly interested in a customized graph dual $G'$ of $G$ that we can use to find cells of minimal visibility.  To obtain $G'$, we first build the dual graph of $G$.  We then remove all edges corresponding to segments of $\partial P$.  This gives us the underlying graph of $G'$, and we will direct the edges as follows.

Each edge of $G$ that does not correspond to part of $\partial P$ separates two cells; thus each edge of $G'$ connects two cells.  Consider an edge $e$ of $G'$ that connects two adjacent cells, or faces, $f_i$ and $f_j$.  If $V_i$ and $V_j$ are the respective sets of vertices that see points in $f_i$ and $f_j$, then either $V_i$ can be obtained from $V_j$ by removing one vertex or vice versa.  We direct $e$ towards the face with the smaller corresponding set, \ie, towards the cell that is less visible.  Doing this for all edges we obtain a directed acyclic graph.  The sinks of this graph correspond to the cells of minimal visibility in $\DDD_V(P,V(P))$.

\begin{figure}
\centering
\includegraphics[width=\linewidth,height=\textheight]{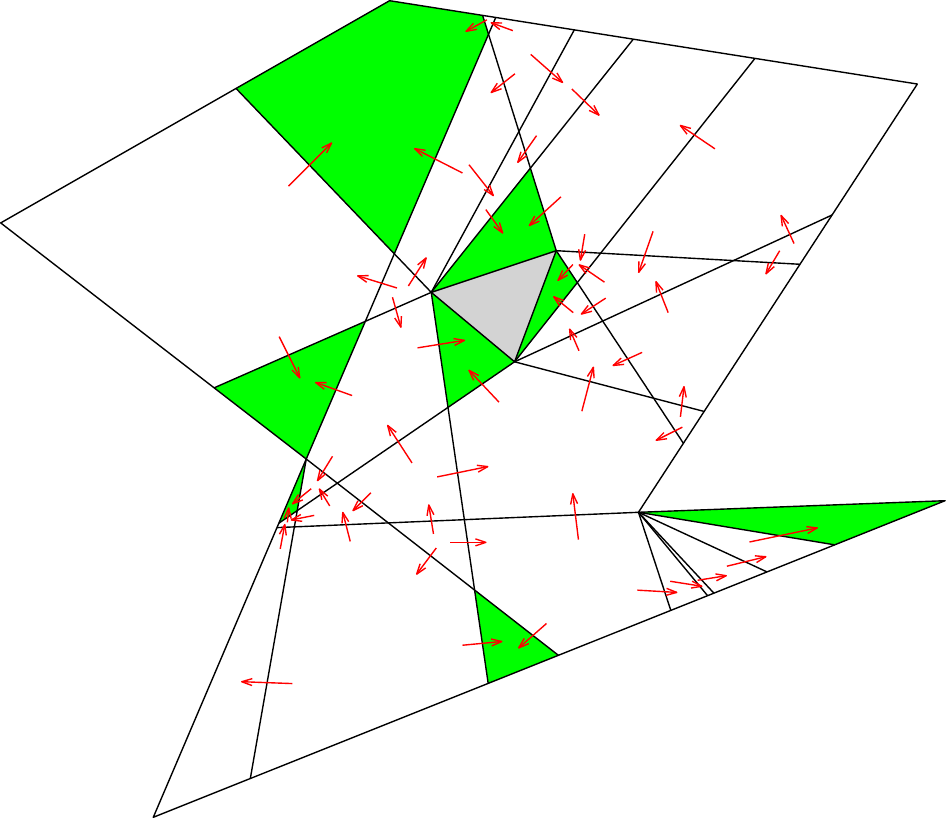}
\caption{The visibility decomposition of a polygon with its directed dual edges indicated in red and its sinks, \ie, cells of minimal visibility, shaded green.\label{fig:decomposition_4}}
\end{figure}

Bose \etal{} used this technique and showed that, given the planar decomposition of a simple polygon $P$, in $\BIGO{n^3}$ time it is possible to construct $G'$ and identify all sinks.  This construction extends naturally to polygons with holes in time $\BIGO{(h+1)n^3}$.

It is worth noting that Bose \etal{} used the planar decomposition algorithm of Bentley and Ottmann \cite{bentley1979}, which in this application is slower by a factor of $\log n$.  They do this because they claim the algorithm of Chazelle and Edelsbrunner \cite{chazelle1992} requires the line segments to be in general position (they will not necessarily be in general position even if the polygon's vertices are in general position).  However, we have looked for and failed to find any original mention of this general position requirement.  At any rate, degeneracies will only occur at the vertices of $P$, and an efficient workaround is possible.

\subsection{Bounding the Number of Sinks}

Here we give an upper bound of $\BIGO{(h+1)^2n^2}$ on the number of sinks.  This generalizes the $\BIGO{n^2}$ bound given by Bose \etal{} for simple polygons.

\begin{lemma}[\cite{bose2002}]\label{lem:bose8}
Given two right pockets $p_1$ and $p_2$, of a point $x$, no point in $p_1$ can see a point in $p_2$.
\end{lemma}
\begin{proof}
The proof of \cite[Lemma 8]{bose2002} works without modification when generalizing to polygons with holes.  For completeness we restate it here.

If a point $a$ in $p_1$ could see a point $b$ in $p_2$, then the line segment between them must intersect both the window of $p_1$ and $p_2$.  Consider the window half planes of $p_1$ and $p_2$.  Since the two half planes intersect at $x$ and both windows are right windows, one of the half planes must contain the windows of the other.  The lemma follows from Lemma \ref{lem:bose1}.
\qed\end{proof}

\begin{lemma}[\cite{bose2002}]\label{lem:bose9}
There is at most one point of intersection between all the right windows of a point $x$ and all the right windows of a distinct point $y$.
\end{lemma}
\begin{proof}
The proof of \cite[Lemma 9]{bose2002} works without modification when generalizing to polygons with holes.  For completeness we restate it here.

Note that if a right window of $x$ intersects a right window of $y$, then both $x$ and $y$ are visible from the intersection point.  Also, from the fact that both are right windows, the base of one window must be contained in the pocket of the other.  There are two cases to consider: either $x$ is contained in a right pocket of $y$ or $x$ is not contained in a right pocket of $y$.  We start with the former.

If $x$ is contained in a right pocket of $y$, then by Lemma \ref{lem:bose8} $x$ cannot see any other right window of $y$, and by Lemma \ref{lem:windows_intersecting_segment}, only one right window of $x$ can intersect the right window of $y$'s right pocket containing $x$.  Therefore, the lemma follows in this case.

Assume that $x$ is not contained in a right pocket of $y$.  Suppose that a right window $r_1$ of $x$ intersects a right window $r_2$ of $y$.  Since $x$ is not conatained in a right pocket of $y$, the base of $r_2$ must be contained in the pocket of $r_1$.  This implies that $y$ is in the window half plane of $r_1$.  However, since $y$ is visible from the intersection point of $r_1$ and $r_2$, $y$ must be in the pocket of $r_1$ by Lemma \ref{lem:bose1}.  Therefore, the lemma follows since $y$ is contained in a right pocket of $x$.
\qed\end{proof}

\begin{lemma}\label{lem:T-windows_and_right_windows}
For distinct vertices $v_i$ and $v_j$, if $v_i$ is not in a right pocket of $v_j$ then any T-window of $v_i$ crosses at most one right window of $v_j$.
\end{lemma}
\begin{proof}
By Lemma \ref{lem:bose8}, a ray shot from $v_i$ cannot leave a right pocket of $v_j$ and enter another right pocket of $v_j$.  The lemma follows easily.
\qed\end{proof}

\begin{corollary}\label{lem:TT_sinks}
The total number of crossings between T-windows is $\BIGO{(h+1)^2n^2}$.
\end{corollary}

\begin{theorem}\label{thm:sinks}
For a polygon $P$ with $h$ holes, $\DDD_V(P,V(P))$ contains $\BIGO{(h+1)^2n^2}$ cells of minimal visibility, and this is tight in the worst case.
\end{theorem}
\begin{proof}
Let $R_i$ be the set of right windows of $v_i$.  We want to show that there are $\BIGO{(h+1)n}$ sinks bordered by a window in $R_i$.  We do this by showing that, for a vertex $v$, there are $\BIGO{h}$ sinks having an edge from a window in $R_i$, followed in clockwise order by an edge from a window of $v$.  If $v$ is in a right pocket of $v_i$, $v$ can only see one window in $R_i$, namely the window bounding the pocket, and by Lemma \ref{lem:windows_intersecting_segment} we are done.  Otherwise, consider a sink $s$ that has one edge formed by a right window $r\in R_i$ and the next edge in clockwise order formed by a window of $v$.  The window of $v$ must be either a right window or a T-window, otherwise $s$ would not be a sink.  By Lemma \ref{lem:bose9}, at most one right window of $v$ can intersect a window in $R_i$.  By Lemma \ref{lem:T-windows_and_right_windows}, each T-window of $v$ crosses at most one window in $R_i$.  

We have shown that there are $\BIGO{(h+1)n}$ sinks bordered by a window in $R_i$.  We can do the same for sinks bordered by windows in $L_i$, the set of left windows of $v_i$.  Thus the total number of sinks bordered by all right and left windows is $\BIGO{(h+1)n^2}$.  By Corollary \ref{lem:TT_sinks} we know that the number of cells having two consecutive sides formed by T-windows is $\BIGO{(h+1)^2n^2}$, since each intersection point of two T-windows borders at most 4 cells.  The upper bound of $\BIGO{(h+1)^2n^2}$ for the total number of cells of minimal visibility follows.

We prove the lower bound by example.  See Figures \ref{fig:sinks_lower_bound_1}, \ref{fig:sinks_lower_bound_2}, \ref{fig:sinks_lower_bound_3}, and \ref{fig:sinks_lower_bound_4}.

\begin{figure}
\centering
\includegraphics[width=\linewidth]{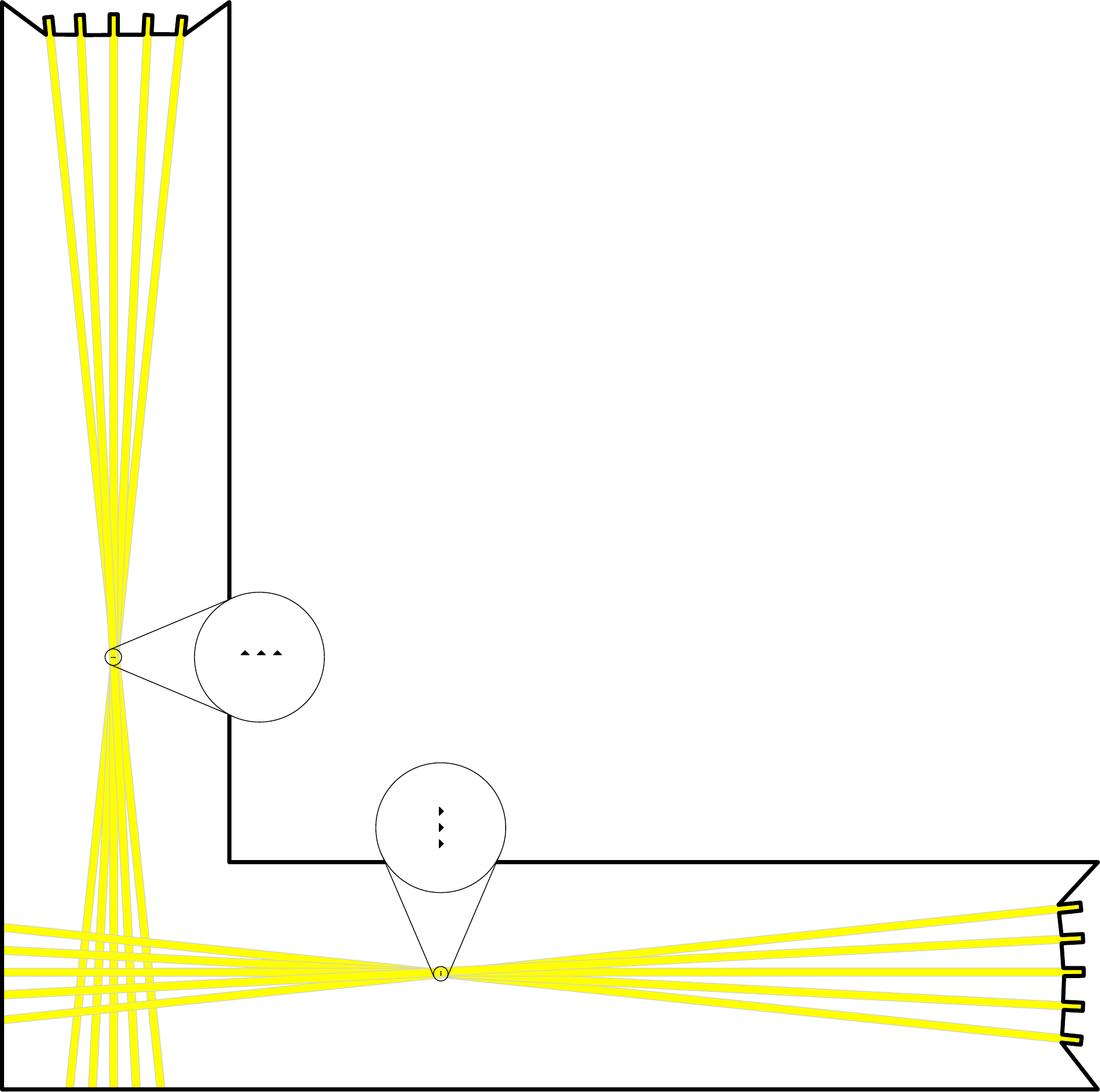}
\caption{A polygon with $\THETA{(h+1)^2n^2}$ sinks.  Each \emph{light strip}, indicated in yellow, is the intersection of visibility polygons of two adjacent vertices.  This polygon has $10 = \THETA{n}$ light strips and 6 holes.  The construction easily generalizes to higher values of $n$ and $h$.  This figure is diagrammatic; in the real polygon the holes would be closer to the region in which horizontal and vertical light strips interact.\label{fig:sinks_lower_bound_1}}
\end{figure}
\begin{figure}
\includegraphics[width=\linewidth]{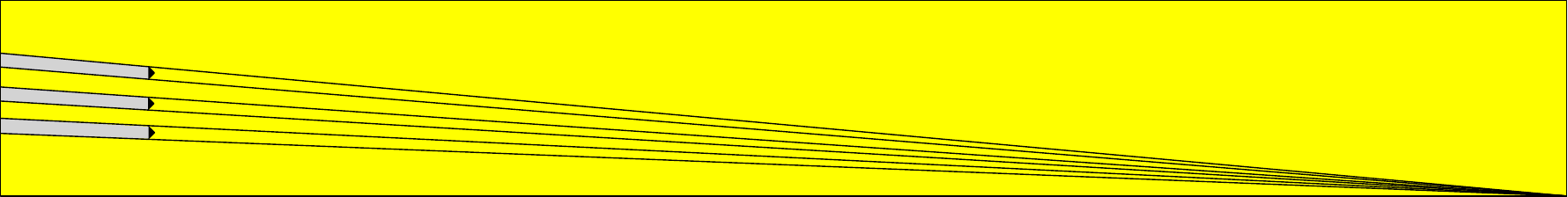}
\caption{Detail of one of the light strips.  The \emph{shadow strips}, indicated in light grey, are the regions not seen by the vertex at the bottom right.  Each light strip has $h/2$ shadow strips.\label{fig:sinks_lower_bound_2}}
\end{figure}

\begin{figure}
\begin{minipage}[t]{0.46\linewidth}
\centering
\includegraphics[width=\columnwidth]{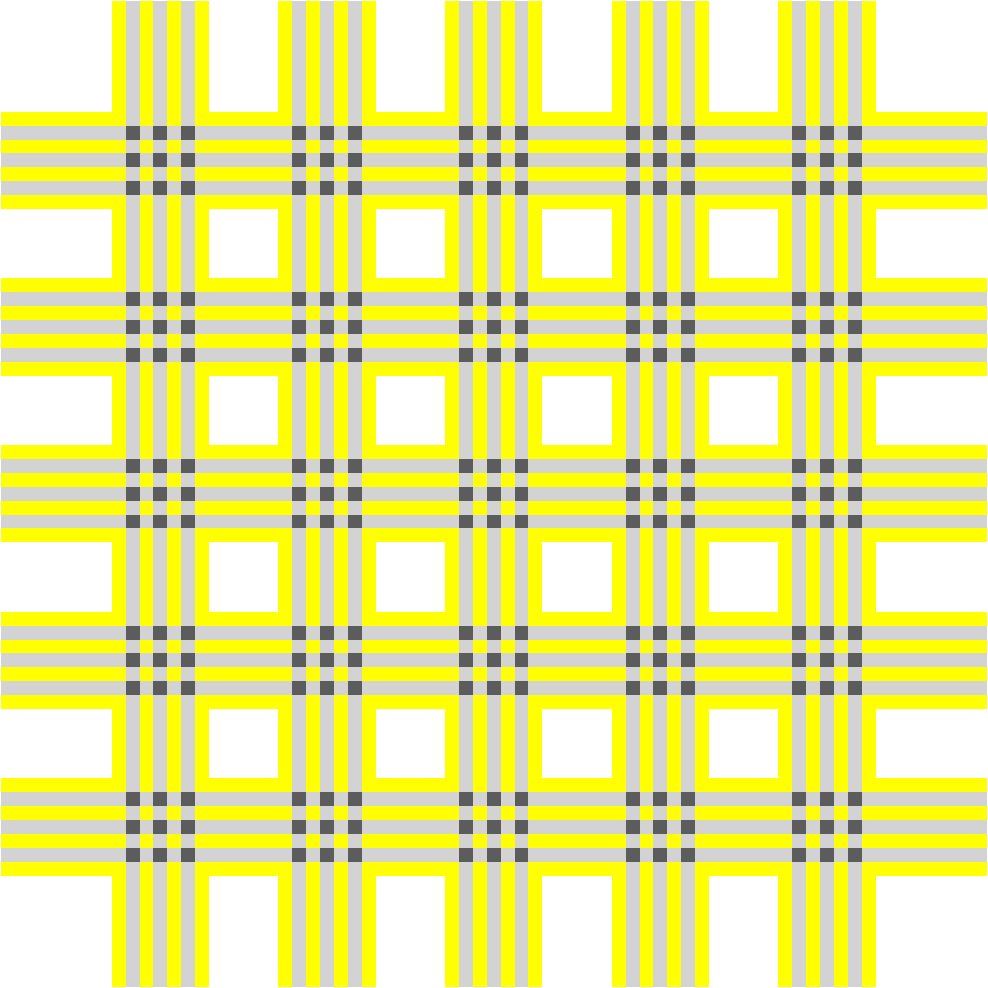}
\caption{Detail of the interaction of shadow strips.  There are $\THETA{(h+1)^2n^2}$ dark grey squares, each of which is the intersection of shadow strips and must contain a sink.\label{fig:sinks_lower_bound_3}}
\end{minipage}
\hspace{0.05\linewidth}
\begin{minipage}[t]{0.46\linewidth}
\centering
\includegraphics[width=\columnwidth]{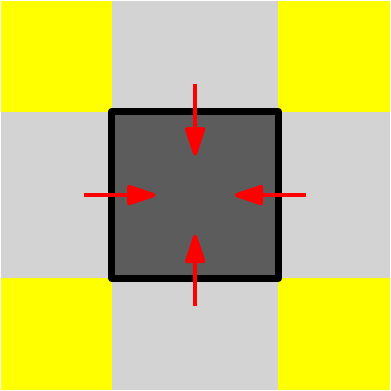}
\caption{Detail of the intersection of shadow strips.  Dual edges are indicated in red; since they are all pointing inwards, the dark grey region must contain a sink.\label{fig:sinks_lower_bound_4}}
\end{minipage}
\end{figure}

\qed
\end{proof}

\section{Greedy Approximation}

Based on the previous section we treat the guarding problem as the abstracted problem of finding a minimum hitting set for a discrete range space $\SSS = (X,\RRR)$, where $|X| = n$ and $|\RRR| = \BIGO{(h+1)^2n^2}$.  This range space can be obtained in time $\BIGO{(h+1)n^3}$.

Using the greedy algorithm for set cover/hitting set on this range space gives us a $\BIGO{\log n}$-approximation algorithm in $\BIGO{|X||\RRR|} = \BIGO{(h+1)^2n^3}$ time.

\hide{
\LinesNumbered
\begin{algorithm}
\SetAlgoLined
\While {there is an unguarded sink}{
  place a guard at the vertex that sees the most unguarded sinks\;\nllabel{line:place}
  \For {each vertex $v$}{
    update the list of sinks that $v$ sees\;
  }
}
\end{algorithm}
The \textbf{while} loop iterates $\BIGO{\opt}$ times.  Line \ref{line:place} can be performed in $\BIGO{n}$ time.
}

\section{Improved Approximation via $\eps$-Nets}

We now turn our attention to the task of achieving a better approximation factor when $\opt$ is small; for our bounds on time complexity we can assume that $\opt = \BIGO{n^{1/3}}$.  Otherwise an approximation factor of $\BIGO{\log n}$ is also $\BIGO{\log \opt}$.

\subsection{VC-Dimension and $\eps$-Nets}

We use standard techniques for approximating hitting set that are based on the concept of \emph{VC-dimension} first introduced by Vapnik and Chervonenkis \cite{vapnik1971} in the area of learning theory.

\begin{definition}[{\bf VC-Dimension} \cite{vapnik1971}]
For a range space $\SSS = (X,\RRR)$, let $Y$ be a maximum cardinality subset of $X$ such that $\RRR\cap Y = 2^Y$.  The \emph{VC-dimension} of $\SSS$ is equal to $|Y|$.
\end{definition}

We use the following bounds on the VC-dimension of visibility systems in polygons due to Valtr \cite{valtr1998}.

\begin{theorem}[\cite{valtr1998}]\label{thm:vc_bounds}
The visibility system of a polygon with $h$ holes has VC-dimension at most $23$ if $h=0$ (\ie, if the polygon is simple), and $2\log_2h + 4\log_2\log_2 h + o(1) = \BIGO{1+\log h}$ if $h\geq 1$.
\end{theorem}

In systems of bounded VC-dimension, small $\eps$-nets can be constructed via random sampling.  We use the following result of Blumer \etal~\cite{blumer1989}.

\begin{theorem}[\cite{blumer1989}]\label{thm:random_sample_eps_net}
For a measure $\mu$ on the elements of a range space of VC-dimension $d$, and for any $\eps\in(0,1]$ and any $\delta\in(0,1]$, a random sample of $m(\eps, \delta)$ elements drawn according to $\mu$ forms an $\eps$-net with probability at least $1-\delta$, where
$$
m(\eps, \delta) = \max\kern-2pt\left( \frac{4}{\eps}\log\frac{2}{\delta} ~,~ \frac{8d}{\eps}\log\frac{13}{\eps} \right)~.
$$
\end{theorem}

The following is a straightforward consequence of Theorems \ref{thm:vc_bounds} and \ref{thm:random_sample_eps_net}.

\begin{corollary}
There exists a function $f(h,\eps) = \BIGO{\Big(1+\log(h+1)\Big)\frac{1}{\eps}\log\frac{1}{\eps}}$ such that, for a measure $\mu$ on a polygon with $h\geq 0$ holes, and for any $\eps\in(0,1]$, a random sample of $f(h,\eps)$ points drawn according to $\mu$ forms an $\eps$-net with probability at least $1/2$.
\end{corollary}

\subsection{Approximation via $\eps$-Nets}

Br\"{o}nnimann and Goodrich \cite{bronnimann1995} developed a method for turning algorithms for finding $\eps$-nets into approximation algorithms for finding minimum hitting sets.  The key is in finding an optimum, or approximately optimum, measure $\mu$ on the elements of a range space.  Their algorithm essentially `learns' a measure $\mu$ through an iterative doubling technique.  The optimum measure $\mu^*$ is the distribution that maximizes the value of $\eps^*$ such that every $\eps^*$-net is a hitting set.  In fact, the size of a minimum fractional hitting set is exactly $\opt_f = 1/\eps^*$ \cite{even2005}; this serves as a lower bound for $\opt$.

\paragraph{Iterative doubling.}
The B\&G algorithm \cite[\S 3.1]{bronnimann1995} finds a measure $\mu'$ such that every $\eps'$-net is a hitting set for some $\eps' \geq \frac{1}{2\cdot\opt}$.  The algorithm starts by assigning a weight of 1 to every element (when the algorithm terminates these weights are normalized to obtain $\mu'$).  In each iteration the algorithm makes one call to the $\eps$-net finder, which returns an $\eps$-net $Y\subseteq X$ and one call to a \emph{verifier}.  The verifier checks if the given set $Y$ is a hitting set.  If so, the algorithm returns $Y$ and terminates.  If not, the verifier returns a range $R$ that is not hit by $Y$; the algorithm doubles the weight of every element in $R$, then starts a new iteration.  The algorithm is guaranteed to terminate after $\BIGO{\opt\cdot\log|X|}$ iterations and the total weight of elements cannot exceed $|X|^4$ \cite{bronnimann1995}.  Since the B\&G algorithm does not know the value of $\opt$ or $\eps'$ \textit{a priori}, it must make several guesses, starting at a constant value such as $1/2$ and halving the guess after each failed run of the algorithm.  The result of this is that the entire algorithm must be run $\BIGO{\log \opt}$ times \cite{bronnimann1995}.

\paragraph{Random sampling and verification.}
We show that, for our application, the total time complexity of random sampling and verification is $\BIGO{n^3}$.

For verification we build a directed bipartite graph with $n$ vertices representing elements and $\BIGO{n^2}$ elements representing ranges, where an element vertex is adjacent to a range vertex if and only if the element is in the range.  This graph is constructed in $\BIGO{n^3}$ time and can be constructed once and used for all iterations of the algorithm.  Each range vertex is given a boolean flag indicating whether or not it is hit by the input set.  In each verification round we do the following:
\begin{compactenum}
\item In $\BIGO{|\RRR|}$ time, reset the flags.
\item In $\BIGO{|Y|\cdot|\RRR|}$ total time, for each element in the input set $Y$, mark all incident ranges as `hit'.
\item In $\BIGO{|\RRR|}$ time, scan through the ranges to find an unhit range or verify that all ranges are hit.
\item In $\BIGO{|\RRR|}$ time, if there is an unhit range, double the weight of every element in that range.
\end{compactenum}
Thus each verification round takes $\BIGO{|Y|\cdot|\RRR|}$ time.

For random sampling, we assume that we can sample a random bit in $\BIGO{1}$ time, allowing us to sample uniformly from an array of $k$ elements in $\BIGO{\log k}$ expected time.  Through all iteratoins of the algorithm, every element weight is a power of 2.  Since the total element weight does not exceed $|X|^4$, each element has its weight doubled $\BIGO{\log |X|}$ times in each run of the B\&G algorithm.  We can maintain a partition of the elements based on how many times each element's weight has been doubled.  In this way we can sample an element according to the weight function in $\BIGO{\polylog |X|}$ amortized time.  The additional cost of maintaining this partition is $\BIGO{|X|\log |X|}$ time per run of the algorithm.

\paragraph{Total time complexity.}
For our application, we have that $Y = \BIGO{n^{1/3}\polylog n}$ and $|\RRR| = \BIGO{(h+1)^2n^2}$.  In each iteration of the B\&G algorithm, the time complexity of sampling for the $\eps$-net finder is 
\begin{eqnarray*}
f(h,\eps)\BIGO{\polylog n} &=& \BIGO{\Big(1+\log(h+1)\Big)\frac{1}{\eps}\log\frac{1}{\eps}\polylog n} \\
&=& \BIGO{\opt \cdot \log \opt \cdot \polylog n} \\
&=& \BIGO{n^{1/3}\polylog n}~.
\end{eqnarray*}
The time complexity of the verifier is $\BIGO{|Y|\cdot|\RRR|} = \BIGO{(h+1)^2n^{7/3}\polylog n}$.  We perform $\BIGO{\log\opt}$ runs of the algorithm and in each run there are \linebreak$\BIGO{\opt\cdot\log n}$ iterations.  Thus the total number of iterations is $\BIGO{\opt\polylog n}$ $= \BIGO{n^{1/3}\polylog n}$.  Our total running time, including all runs of the algorithm, all iterations of the $\eps$-net finder and verifier, and additional overhead, does not exceed $\BIGO{(h+1)^2n^3}$.

\paragraph{Approximation ratio.}
When the last run of the algorithm terminates with success, we are left with a measure $\mu'$ on $X$, along with an $\eps'$-net that is a hitting set, where $\eps' \geq \frac{1}{2\cdot\opt}$.  The size of this hitting set is at most 
\begin{eqnarray*}
f(h, \eps') &=& \BIGO{\Big(1+\log(h+1)\Big)\frac{1}{\eps'}\log\frac{1}{\eps'}} \\
&=& \BIGO{\Big(1+\log(h+1)\Big)\opt\cdot\log\opt}~.
\end{eqnarray*}
Thus we have shown the following theorem.
\begin{theorem}\label{thm:main-withholes}
For simple polygons or polygons with the number of holes bounded by a constant, there exists an approximation algorithm running in $\BIGO{n^3}$ time with an approximation ratio of $\BIGO{\log\opt}$.  For polygons with $h\geq 2$ holes, there exists an approximation algorithm running in $\BIGO{h^2n^3}$ time with an approximation ratio of $\BIGO{\log h\log \opt}$.
\end{theorem}

\section{Further Improved Approximation for Simple Polygons}

For simple polygons, King and Kirkpatrick \cite{king2010loglog} presented an $\eps$-net finder that returns $\eps$-nets of size $\BIGO{\frac{1}{\eps}\log\log\frac{1}{\eps}}$; they did not analyze its time complexity, rather they simply stated that it runs in polynomial time.  Here we show that it runs in time $\BIGO{n^2\log\log\frac{1}{\eps}}$.  This means that, for simple polygons, we can use it to replace the random sampling $\eps$-net finder of the previous section without pushing the total running time above $\BIGO{n^3}$.

First we note that the B\&G algorithm always uses values of $\eps$ that are powers of 2.  This means that 
$$
2^{{2}^{\ceil{\log\log 1/\eps}}} = \BIGO{1/\eps}~,
$$which slightly facilitates analysis.  With this restriction on $\eps$, the K\&K net finder recursively partitions the vertex set of the polygon, keeping its cyclic ordering.  Define $t = \ceil{\log\log t}$.  At the $i^{\mathrm{th}}$ level of the partition, each subset is further divided into $b_i$ subsets, with
$$
b_i = \left\{
\begin{array}{lcl}
2^{2^{t-1}+1}\cdot 4t\cdot 2^{1-t} &,&i=1\\
2^{2^{t-i}+1}&,&1<i\leq t~.
\end{array}
\right.
$$
If $f_i$ is the number of new fragments created by the $i^{\rm th}$ fragmentation step, this gives us
$$
f_i = \left\{
\begin{array}{ccl}
1 &,&i=0\\
4t\cdot 2^{2^t - 2^{t-i} -t +i +1}&,&0<i\leq t \\
4t\cdot 2^{2^t}&,&i = t~.
\end{array}
\right.
$$
In $\BIGO{n\log\log\frac 1 \eps}$ time we can build a tree corresponding to the hierarchical decomposition in which each node stores the corresponding subset in cyclic order.  This tree has $\BIGO{\frac 1 \eps \log\log \frac 1 \eps}$ nodes.  Determining guards to place for a pair of sibling fragments, respectively storing vertex sets $U_1$ and $U_2$, can be done with $\BIGO{|U_1||U_2|}$ calls to a visibility oracle.  The visibility matrix of the vertices of $P$ can easily be built in $\BIGO{n^3}$ time and serves as a constant-time oracle.

Each of the $f_{i-1}$ nodes at level $i-1$ contains $\BIGO{n/f_{i-1}}$ vertices; it has $b_i$ normal children, each with $\BIGO{n/f_i}$ vertices, and one dummy child with $\BIGO{n}$ vertices.  The cost of placing all guards at level $i$ is therefore
\begin{eqnarray*}
\BIGO{f_{i-1}\left( b_i^2 \left(\frac{n}{f_i}\right)^2 + b_i \left(n\cdot\frac{n}{f_i}\right)\right)}
&=&
\BIGO{\frac{n^2b_if_{i-1}}{f_i}\left( \frac{b_i}{f_i} + 1 \right)} \\
&=&
\BIGO{n^2\left( \frac{b_i}{f_i} + 1 \right)} \\
&=&
\BIGO{n^2}~.
\end{eqnarray*}
Thus the cost of placing all guards at all levels is $\BIGO{n^2\log\log \frac{1}{\eps}}$, as desired.

Using this $\eps$-net finder instead of random sampling, we can achieve an approximation ratio of $\BIGO{\log\log\opt}$ for simple polygons in $\BIGO{n^3}$ time.

\begin{theorem}\label{thm:main-loglog}
For simple polygons, there exists an approximation algorithm running in $\BIGO{n^3}$ time with an approximation ratio of $\BIGO{\log\log\opt}$.
\end{theorem}

\bibliographystyle{splncs03}
\bibliography{./fast_vertex_guarding}

\end{document}